\documentclass[12pt]{article}
\usepackage{amsmath}
\usepackage{amssymb}
\usepackage{amsthm}
\usepackage{epsfig}
\usepackage{graphicx}
\usepackage{enumerate}
\usepackage{cite}
\usepackage{color}

\newcommand{\numberset}{\mathbb}

\newcommand{\F}{\numberset{F}}

\newcommand{\X}{\mathcal{X}}
\newcommand{\C}{\mathcal{C}}
\newcommand{\HF}{\mathcal{H}}
\newcommand{\V}{\mathcal{V}}
\newcommand{\LR}{\mathcal{L}}
\newcommand{\D}{\mathcal{D}}
\newcommand{\A}{\mathcal{A}}
\newcommand{\B}{\mathcal{B}}
\newcommand{\CF}{\mathcal{C}}

\newtheorem{theorem}{Theorem}[section]

\newtheorem{corollary}[theorem]{Corollary}

\newtheorem{definition}[theorem]{Definition}
\newtheorem{remark}[theorem]{Remark}

\title{Constant dimension codes from Riemann-Roch spaces}
\date{\today}
\author{Daniele Bartoli
\thanks{The first author acknowledges the support of the European Community under a Marie-Curie Intra-European Fellowship (FACE project: number 626511).}\\
Department of Mathematics, Ghent University,\\
Krijgslaan 281, 9000 Ghent, Belgium\\
\vspace*{0.5 cm}\\
Matteo Bonini\\
Department of Mathematics and Computer Science,\\
Universit\`a degli Studi di Perugia,\\
Via Vanvitelli 1, 06123, Perugia, Italy\\
\vspace*{0.5 cm}\\
Massimo Giulietti\\
Department of Mathematics and Computer Science,\\
Universit\`a degli Studi di Perugia,\\
Via Vanvitelli 1, 06123, Perugia, Italy\\}

\begin{document}
\maketitle

\begin{abstract} 
Some families of constant dimension codes arising from Riemann-Roch spaces associated to particular divisors of a curve $\X$ are constructed. These families are  generalizations of the one constructed by Hansen
 \cite{hansen}.
\end{abstract}

\section{Introduction}
Let $V=\F_q^N$ be an $N$-dimensional vector space over $\F_q$, $q$ any prime power. The set $\mathcal{P}(V)$ of all subspaces of $V$ forms a metric space with respect to the {\em subspace distance} defined by $d_s(U,U')= \dim (U+U') - \dim (U \cap U')$; see \cite{KK}. 
In this general setting,  a \emph{subspace code} $\C$ is a subset of the set $\mathcal{P}(V)$. Moreover, if all the subspaces of $\C$ have a fixed dimension $\ell$, then $\C$ is called \emph{constant dimension code} (or Grassmannian code) and $\C$ is a subset of $G(\ell,N)(\F_q)$ the set of all the $\ell$-dimensional subspaces of $V=\F_q^{N}$. 
Recently, there has been a lot of interest in codes whose codewords are vector subspaces of a given vector space over $\F_q$, since they have been proposed for error control in random linear network coding; see \cite{KK}. For general results on bounds and constructions of constant--dimension subspaces codes, see \cite{CP1,ES1,ES,EV,GadouleauYan,HKK,KSK,SKK,TR, BP}.

In this paper we describe some families of constant dimension codes arising from algebraic curves over finite fields. Namely, the codewords of these codes will be Riemann-Roch spaces associated to particular divisors. 
The families we will present are a generalization of the one presented by Hansen; see \cite{hansen}.

\section{Hansen's construction}
First of all we recall the definition of constant dimension codes and the related parameters. 
\begin{definition}
A \emph{constant dimension code}  $\C\subseteq G(\ell,N)(\F_q)$ is a set of $\ell$-dimensional $\F_q$-linear subspaces of $\F_q^N$. The \emph{size} of the code is denoted by $|\C|$ and the \emph{minimum distance} by 

\[
D(\C):=\min_{V_1,V_2\in{}\C,V_1\neq{}V_2} d_s(V_1,V_2)
\]

The linear network code $C$ is said to be of type $[N,\ell,\log_q|\C|,D(\C)]$. Its \emph{normalized weight} is $\lambda(\C)=\frac{\ell}{N}$, its \emph{rate} is $R(\C)=\frac{\log_q(|\C|)}{Nl}$ and its \emph{normalized minimal distance} is $\delta(\C)=\frac{D(\C)}{2\ell}$.
\end{definition}

Now we present the construction due to Hansen; see \cite{hansen}.

Let  $\X$ be an absolutely irreducible, projective algebraic curve of genus $g$ defined over  $\F_q$ and $X(\F_q)$  the set of the $\F_q$-rational places of $\mathcal{X}$. Also, let  $n=|\mathcal{X}(\F_q)|$. Fix a positive integer $k$ and consider  $k\sum_{P\in \X(\F_q)}=k\D$, the Frobenius invariant divisor of degree $kn$ having as support the set of all of the $\F_q$-rational places of $\X$. The ambient vector space $\overline{W}$ of this family of  linear network codes will be  the Riemann-Roch space
\[
\overline{W}=\LR\left(k\sum_{P\in{}\mathcal{X}(\F_q)}P\right).
\]
If $nk>2g-2$ from the Riemann-Roch theorem we have that $\text{dim}_{\F_q}\overline{W}=kn+1-g=N$.

Let $s$ be a fixed non-negative integer. The family $\HF$ of linear network codes presented in \cite{hansen} is defined as follows.

\begin{definition}\label{def:Hansen}
Let $\V_s=\left\{ \sum_{P\in S}P \ | \ S\subseteq \X(\F_q), |S|=s\right\}$.
The family $\HF$ is given by
\begin{equation}
\label{codicehansen}
\HF_{k,s}=\left\{\LR\left(kV\right) \ | \ V \in \V_{s}\right\}.
\end{equation}
\end{definition}

Since each divisor in $\V_s$ has degree $s$,  by the Riemann-Roch theorem, if $ks>2g-2$ then each codeword of $\HF_{k,s}$ has dimension $ks+1-g$. 

Hansen \cite{hansen} determined the parameters of the code $\HF_{k,s}$. We summarize  its results in the following theorem.

\begin{theorem}[Hansen, \cite{hansen}]
\label{hansenteo}
Let $\HF_{k,s}$ be the linear network code as in Definition \rm\ref{def:Hansen}. 

Assume $k, s$ be positive integers satisfying $ks>2g-2$.

Then $\HF_{k,s}$ is a $\left[kn+1-g, ks+1-g,\log_q \binom{n}{s},D(\HF_{k,s})\right]$, where
$$D(\HF_{k,s}) =\left\{ \begin{array}{ll} 2k, & s=1;\\ 2(k+1-g),& s>1.\\ \end{array}\right.$$
Also, normalized weight, rate, and normalized minimal distance are 
$$\lambda(\HF_{k,s})=\frac{ks+1-g}{nk+1-g}, \qquad R(\HF_{k,s})=\frac{\text{log}_q(\binom{n}{s})}{(nk+1-g)(ks+1-g)},$$
$$\delta(\HF_{k,s})=\frac{1}{s+\frac{1-g}{k}}\ge\frac{2g-1}{(s+1)g-1}.$$
\end{theorem}

\section{Some generalizations}

We generalized the family of linear network codes $\HF_{k,s}$, basically by considering sets of divisors of fixed degree $s$ of size larger than $|\V_{s}|$ (see Definition \ref{def:Hansen}). In this section we present three families, which can be seen as a generalizations of $\HF_{k,s}$. 

\subsection{The family $\A_{k,s}$}
We consider divisors of fixed degree $s$ having non-negative weights.

\begin{definition}\label{def:firstfamily}
Let $k,s$ be positive integers. Let 
$$\V_s^{\prime}=\left\{ \sum_{P\in \X(\F_q)}m_{P} P \ \Big| \  \sum_{P\in \X(\F_q)}m_{P}=s, \quad m_P \in \{0,\ldots,s\}\right\}.$$
The family $\A_{k,s}$ is given by
\begin{equation}
\label{codicehansen}
\A_{k,s}=\left\{\LR\left(kV\right) \ | \ V \in \V_{s}^{\prime}\right\}.
\end{equation}
\end{definition}
Note that in this case the ambient space is larger than in the case of family $\HF_{k,s}$, since each codeword of  $\A_{k,s}$ is contained in $W=\LR\left( ks\D\right)$. Also, if $nks>2g-2$, by the Riemann-Roch theorem, $\text{dim}_{\F_q}\left(\LR\left( ks\D\right)\right)=nks+1-g$.

\begin{theorem}
\label{th:family1}
Let $\A_{k,s}$ be the linear network code as in Definition \rm\ref{def:firstfamily}. 

Assume $k, s$ be positive integers satisfying $k>2g-2$.

Then $\A_{k,s}$ is a $\left[nks+1-g, ks+1-g,\log_q \binom{n+s-1}{s},D(\A_{k,s})\right]$, where
$$D(\A_{k,s}) =\left\{ \begin{array}{ll} 2k, & s=1;\\ 2(k+1-g),& s>1.\\ \end{array}\right.$$
Also, normalized weight, rate, and normalized minimal distance are 
$$\lambda(\A_{k,s})=\frac{ks+1-g}{nks+1-g}, \qquad R(\A_{k,s})=\frac{\text{log}_q(\binom{n+s-1}{s})}{(nks+1-g)(ks+1-g)},$$
$$\delta(\A_{k,s})=\frac{1}{s+\frac{1-g}{k}}\ge\frac{2g-1}{(s+1)g-1}.$$
\end{theorem}
\begin{proof}

By our assumptions $k>2g-2$, which implies $nks>2g-2$ and therefore $\dim_{\F_q} W=nks+1-g$. Also, each codeword of $\A_{k,s}$ has dimension over $\F_q$ equal to $ks+1-g$. The number of codewords is exactly the number of solutions of the linear equation
$$x_1+x_2+\cdots+x_n=s,$$
where $x_i \in \{0,\ldots,s\}$. It also corresponds to the number of $s$-combinations with repetitions of $n$ elements, namely $\binom{n+s-1}{s}$.

In order to compute the minimum distance of this code, first note that for any two divisors $V_1$ and $V_2$ in $\V_s^{\prime}$, with  $$V_1=\LR\left(k\sum_{P \in \X(\F_q)} m_P P\right), \qquad V_2=\LR\left(k\sum_{P \in \X(\F_q)} \overline{m}_P P\right),$$ and therefore 

\begin{equation}
\label{intersezione}
V_1\cap V_2=\LR\left(k\sum_{P \in \X(\F_q)} \min\{m_P,\overline{m}_P\} P\right).
\end{equation}
This implies that if $k>2g-2$ then 
$$
\dim(V_1\cap{}V_2)=k\sum_{P \in \X(\F_q)} \min\{m_P,\overline{m}_P\} +1-g.
$$
If $s=1$, then the intersection between two different spaces $V_1$ and $V_2$ has dimension $0$. From the definition of the metric we have that:

$$
d_s(V_1,V_2)=2(k+1-g),\quad \mathrm{for\  all}\ V_1,V_2\in\A_{k,1}.
$$

Consider now $s>1$ and let $V_1$ and $V_2$ be two distinct codewords. Therefore there exists a place $P\in \X(\F_q)$ such that $m_P\neq \overline{m}_P$. This implies that $\dim(V_1\cap{}V_2)=k\sum_{P \in \X(\F_q)} \min\{m_P,\overline{m}_P\} +1-g\leq k(s-1)+1-g$.
Recalling that 
$d_s(V_1,V_2)=2\ell-2\dim(V_1\cap V_2)$, we obtain
$$d_s(V_1,V_2) \geq 2(ks+1-g)-2(k(s-1)+1-g)=2k.$$
Finally, note that the following two codewords $V_1=\LR(k(s-1)P+kQ)$ and $V_2=\LR(k(s-1)P+kR)$, with $P,Q,R$ pairwise distinct places of $\X$, have distance equal to $2k$. This means that the minimum distance of the code $\A_{k,s}$ is exactly $2k$. 

Concerning the normalized weight, rate, and normalized minimal distance, their computations are straightforward. The estimate on $\delta(\A_{k,s})$  follows from the fact that $k\geq 2g-1\geq \frac{2g-1}{s-1}$. 
\end{proof}

\begin{remark}
The assumption $k>2g-2$ is necessary to know the exact dimension of $V_1\cap{}V_2$, since otherwise the Riemann-Roch theorem would imply only $
\dim(V_1\cap{}V_2)\geq k\sum_{P \in \X(\F_q)}\min\{m_P,\overline{m}_P\}+1-g.$
\end{remark}

We can observe that  in the construction of codes $\HF_{k,s}$ the divisors in $\V_s$ correspond to $s$-subsets of the set of all the $\F_q$-rational places of $\X$; here the divisors in $\V_s^{\prime}$ are in bijection with the $s$-multisubsets of $\X(\F_q)$.
This shows the first difference between $\HF_{k,s}$ and $\A_{k,s}$. In fact, in the first case the parameter $s$ can be at most $n$, whereas in the second case we can allow $s$ to be greater than $n$. So, in principle the construction $\A_{k,s}$ can be also applied to curves $\X$ not having a large number of $\F_q$-rational places.

\subsection{The family $\B_{k,s}$}
In this case we consider divisors of fixed degree $s$ having non-negative weights bounded by another constant $w$. In the case $w=s$ this new family $\B_{k,s,s}$ coincides with $\A_{k,s}$. The purpose of this  generalization is to bound the dimension of the ambient space.

\begin{definition}\label{def:secondfamily}
Let $k,s,w$ be fixed positive integers, with $0<w\leq s\leq nw$. Let 
$$\V_s^{\prime\prime}=\left\{ \sum_{P\in \X(\F_q)}m_{P} P \ \Big| \  \sum_{P\in \X(\F_q)}m_{P}=s, \quad m_P \in \{0,\ldots,w\}\right\}.$$
The family $\B_{k,s,w}$ is given by
\begin{equation}
\label{codicehansen}
\B_{k,s,w}=\left\{\LR\left(kV\right) \ | \ V \in \V_{s}^{\prime\prime}\right\}.
\end{equation}
\end{definition}

In order to compute the number of codewords of the code $\B_{k,s,w}$ we will use the following result.
\begin{theorem}[Wu, \cite{wu}]
\label{wu}
Let $n,s,w$ be non-negative integers satisfying $0<w\leq s\leq nw$. 
The number of solutions of the linear equation
$$
x_1+\dots+x_n=s, \quad x_i\in \{0,\ldots, w\},
$$
is
\begin{equation}\label{def_U}
U_{n,s,0,w}=\sum_{i=0}^{t}(-1)^i\binom{n}{i}\binom{s-i(w+1)+n-1}{n-1},
\end{equation}
where $t=\min\left(n,\left\lfloor\frac{s}{b+1}\right\rfloor\right)$.
\end{theorem}

The following theorem describes the parameters of the codes of family $\B_{k,s,w}$. The proof is very similar to the proof of Theorem \ref{th:family1} and therefore we omit it. We used Theorem \ref{wu} in order to compute the number of codewords in $\B_{k,s,w}$.
\begin{theorem}
\label{th:family2}
Let $\B_{k,s,w}$ be the linear network code as in Definition \rm\ref{def:secondfamily}. 

Assume $k, s,w$ are positive integers satisfying $k>2g-2$ and $0<w\leq s\leq nw$.

Then $\B_{k,s,w}$ is a $\left[nkw+1-g, ks+1-g,\log_q U_{n,s,0,w},2k\right]$, where $U_{n,s,0,w}$ is defined in Equation \eqref{def_U}.
Also, normalized weight, rate, and normalized minimal distance are 
$$\lambda(\B_{k,s,w})=\frac{ks+1-g}{nkw+1-g}, \qquad R(\B_{k,s,w})=\frac{\log_q U_{n,s,0,w} }{(nkw+1-g)(ks+1-g)},$$
$$\delta(\B_{k,s,w})=\frac{1}{s+\frac{1-g}{k}}\ge\frac{2g-1}{(s+1)g-1}.$$
\end{theorem}

\subsection{The family $\CF_{k,s,w}$}

Our last generalization takes into account the fact that allowing the divisors of the fixed degree $s$ to have also negative weights increases the number of codewords without changing the dimension of the ambient space. In order to compute the parameters of this new family we need the following corollary to Theorem \ref{wu}.

\begin{corollary}
\label{cor<>}
Let $a\leq 0\leq b$ be two integers, satisfying $b\leq s-b(n-1)$. The number of solution of the diophantine equation 
\begin{equation}\label{eq:sistema}
x_1+\dots+x_n=s,\qquad x_i \in \{a,\ldots,b\} \quad \forall \ i\in \{1,\dots,n\}
\end{equation}
is given by 
\begin{equation}\label{def:UU}
U^{\prime}_{n,s,a,b}=U_{n,s-na,0,b-a}=\sum_{i=0}^{t}(-1)^i\binom{n}{i}\binom{s-na-i(b-a+1)+n-1}{n-1},
\end{equation}
where $t=\min\left(n,\lfloor\frac{s-na}{b-a+1}\rfloor\right)$.
\end{corollary}
\begin{proof}
First note that Equation \eqref{eq:sistema} is equivalent to 
\[
y_1+\dots+y_n=s-na,\quad y_i\in \{0, b-a\}\,\,\forall \ i\in \{1,\dots,n\}.
\]
By Theorem \ref{wu}, the number of solutions of this last equation is 
\begin{equation*}
U^{\prime}_{n,s,a,b}=\sum_{i=0}^{t}(-1)^i\binom{n}{i}\binom{s-na-i(b-a+1)+n-1}{n-1},
\end{equation*}
where $t=\min\left(n,\lfloor\frac{s-na}{b-a+1}\rfloor\right)$.
\end{proof}

\begin{definition}\label{def:thirdfamily}
Let $k,s,w$ be fixed positive integers, with $0<w\leq s\leq nw$. Let 
$$\V_s^{\prime\prime\prime}=\left\{ \sum_{P\in \X(\F_q)}m_{P} P \ \Big| \  \sum_{P\in \X(\F_q)}m_{P}=s, \quad m_P \in \{s-w(n-1),\ldots,w\}\right\}.$$
The family $\CF_{k,s,w}$ is given by
\begin{equation}
\label{codicehansen}
\CF_{k,s,w}=\left\{\LR\left(kV\right) \ | \ V \in \V_{s}^{\prime\prime\prime}\right\}.
\end{equation}
\end{definition}
\begin{remark}
In Definition \ref{def:thirdfamily} we restrict ourself to the case  $m_P \in \{s-w(n-1),\ldots,w\}$ since if for some $\overline{P} \in \X(\F_q)$ such that $m_{\overline{P}}<s-w(n-1)$, then $\sum_{P\in \X(\F_q)}m_{P} = \sum_{P\neq \overline{P}\in \X(\F_q)}m_{P}+m_{\overline{P}}<(n-1)w+s-w(n-1)=s$.
\end{remark}

\begin{theorem}
\label{th:family3}
Let $\CF_{k,s,w}$ be the linear network code as in Definition \rm\ref{def:thirdfamily}. 

Assume $k, s,w$ are positive integers satisfying $k>2g-2$ and $0<w\leq s\leq nw$.

Then $\CF_{k,s,w}$ is a $\left[nkw+1-g, ks+1-g,\log_q U^{\prime}_{n,s,s-w(n-1),w},2k\right]$, where $U^{\prime}_{n,s,s-w(n-1),w}$ is defined in Equation \eqref{def:UU}.
Also, normalized weight, rate, and normalized minimal distance are 
$$\lambda(\CF_{k,s,w})=\frac{ks+1-g}{nkw+1-g}, \qquad R(\CF_{k,s,w})=\frac{\log_q U^{\prime}_{n,s,s-w(n-1),w} }{(nkw+1-g)(ks+1-g)},$$
$$\delta(\CF_{k,s,w})=\frac{1}{s+\frac{1-g}{k}}\ge\frac{2g-1}{(s+1)g-1}.$$
\end{theorem}

\begin{proof}
The proof is very similar to those of Theorem \ref{th:family1} \and Theorem \ref{th:family2}. We note that in this case $U^{\prime}_{n,s,s-w(n-1),w}$ reads 
$$\sum_{i=0}^{t}(-1)^i\binom{n}{i}\binom{(nw-s+1)(n-i-1)}{n-1},
$$
where $t=\min\left(n,\left\lfloor\frac{(nw-s)(n-1)}{nw-s+1)}\right\rfloor\right)=\left\lfloor\frac{(nw-s)(n-1)}{nw-s+1}\right\rfloor$.

\end{proof}

\section{Some comparisons}
In this section we present some computations on the rates of the three  families described in the paper and of the family $\HF_{k,s}$ (see Definition \ref{def:Hansen}). Due to the shape of the formula
$\log_q U^{\prime}_{n,s,s-w(n-1),w}$ in Theorem \ref{th:family3}, we gave some restrictions on the values of the parameters $n,s,w$, in order to handle it. In Table \ref{table:1} we summarize the normalized weight, the rate, and the normalized minimal distance of the four families.

In particular, we focused on their rates. Also,  we consider curves $\X$ of genus $1$: this simplifies the formulas, as shown in Table \ref{table:2}. Due to the difficulty of the approximation of the quantity defined in Formulas \eqref{def_U} and \eqref{def:UU} we could produce the exact values of rates just for small values of $n=|\X(\F_q)|$ and the parameter $s$. These results, for $q=16$, $n=15$, $1\leq s <n$, $w=3$, $k=5$, are summarized in Table \ref{table:3}. Note that in many cases the rate of the third family $\CF_{k,s,w}$ is larger than the rate of $\HF_{k,s}$. It is worth noting that asymptotic formulas for \eqref{def_U} and \eqref{def:UU} would help for the comparisons of the rates of the four families. Finally, note that whereas the parameter $s$ in $\HF_{k,s}$ is upper bounded by $n=|\X(\F_q)|$, in the three families $\A_{k,s}$, $\B_{k,s,w}$, $\CF_{k,s,w}$ we can always consider $s>n$ too.

\begin{table}
\caption{Normalized weight, rate, and  normalized minimal distance}\label{table:1}
\begin{center}
\begin{tabular}{|c|c|c|c|}
\hline
&&&\\[-3 mm]
&\begin{tabular}{c}
{Normalized}\\ 
{weight}\\ 

\end{tabular}
&
\begin{tabular}{c}
{Rate}\\

\end{tabular}
&
\begin{tabular}{c}
{Normalized}\\
{minimum}\\
{distance}\\
\end{tabular}\\
&&&\\[-3 mm]
\hline
&&&\\[-3 mm]
$\HF_{k,s}$&$\frac{ks+1-g}{nk+1-g}$&$\frac{\log_q \binom{n}{s}}{(nk+1-g)(ks+1-g)}$&$ \frac{1}{s+\frac{1-g}{k}}$\\
&&&\\[-3 mm]
\hline
&&&\\[-3 mm]
$\A_{k,s}$&$\frac{ks+1-g}{nks+1-g}$&$\frac{\log_q \binom{n+s-1}{s}}{(nks+1-g)(ks+1-g)}$&$\frac{1}{s+\frac{1-g}{k}}$\\
&&&\\[-3 mm]
\hline
&&&\\[-3 mm]
$\B_{k,s,w}$&$\frac{ks+1-g}{nkw+1-g}$&$\frac{\log_q U_{n,s,0,w}}{(nkw+1-g)(ks+1-g)}$&$\frac{1}{s+\frac{1-g}{k}}$\\
&&&\\[-3 mm]
\hline
&&&\\[-3 mm]
$\CF_{k,s,w}$&$\frac{ks+1-g}{nkw+1-g}$&$\frac{\log_q U^{\prime}_{n,s,s-w(n-1),w}}{(nkw+1-g)(ks+1-g)}$&$\frac{1}{s+\frac{1-g}{k}}$\\[2 mm]

\hline
\end{tabular}
\end{center}
\end{table}

\begin{table}
\caption{Normalized weight, rate, and  normalized minimal distance for $g=1$}\label{table:2}
\begin{center}
\begin{tabular}{|c|c|c|c|}
\hline
&&&\\[-3 mm]
&\begin{tabular}{c}
{Normalized}\\ 
{weight}\\ 

\end{tabular}
&
\begin{tabular}{c}
{Rate}\\

\end{tabular}
&
\begin{tabular}{c}
{Normalized}\\
{minimum}\\
{distance}\\
\end{tabular}\\
&&&\\[-3 mm]
\hline
&&&\\[-3 mm]
$\HF_{k,s}$&$\frac{s}{n}$&$\frac{\log_q \binom{n}{s}}{nk^2s}$&$ \frac{1}{s}$\\
&&&\\[-3 mm]
\hline
&&&\\[-3 mm]
$\A_{k,s}$&$\frac{1}{s}$&$\frac{\log_q \binom{n+s-1}{s}}{nk^2s^2}$&$\frac{1}{s}$\\
&&&\\[-3 mm]
\hline
&&&\\[-3 mm]
$\B_{k,s,w}$&$\frac{s}{nw}$&$\frac{\log_q U_{n,s,0,w}}{nk^2ws}$&$\frac{1}{s}$\\
&&&\\[-3 mm]
\hline
&&&\\[-3 mm]
$\CF_{k,s,w}$&$\frac{s}{nw}$&$\frac{\log_q U^{\prime}_{n,s,s-w(n-1),w}}{nk^2ws}$&$\frac{1}{s}$\\[2 mm]

\hline
\end{tabular}
\end{center}
\end{table}

\begin{table}
\caption{Rates of $\HF_{k,s}$, $\A_{k,s}$, $\B_{k,s,w}$, $\CF_{k,s,w}$ for $q=16$, $8\leq n\leq 14$, $1\leq s <n$, $w=3$, $k=5$}\label{table:3}
\begin{center}
{\tiny
\tabcolsep = 1 mm
\begin{tabular}{|c|c|c|c|c|}
\hline
$(n,s)$& $\HF_{k,s}$ & $\A_{k,s}$ & $\B_{k,s,w}$ & $\CF_{k,s,w}$\\
\hline
$(8,1)$ & $0.003750$  & $0.003750$ & $0.001250$ & $0.008732$\\ 
$(8,2)$ & $0.003005$  & $0.001616$ & $0.001077$ & $0.004286$\\ 
$(8,3)$ & $0.002420$  & $0.000959$ & $0.000959$ & $0.002802$\\ 
$(8,4)$ & $0.001915$  & $0.000654$ & $0.000868$ & $0.002058$\\ 
$(8,5)$ & $0.001452$  & $0.000481$ & $0.000792$ & $0.001611$\\ 
$(8,6)$ & $0.001002$  & $0.000373$ & $0.000728$ & $0.001311$\\ 
$(8,7)$ & $0.000536$  & $0.000300$ & $0.000671$ & $0.001095$\\ 
$(9,1)$ & $0.003522$  & $0.003522$ & $0.001174$ & $0.008931$\\ 
$(9,2)$ & $0.002872$  & $0.001526$ & $0.001017$ & $0.004394$\\ 
$(9,3)$ & $0.002368$  & $0.000909$ & $0.000909$ & $0.002880$\\ 
$(9,4)$ & $0.001938$  & $0.000622$ & $0.000826$ & $0.002121$\\ 
$(9,5)$ & $0.001551$  & $0.000459$ & $0.000758$ & $0.001665$\\ 
$(9,6)$ & $0.001184$  & $0.000357$ & $0.000700$ & $0.001360$\\ 
$(9,7)$ & $0.000821$  & $0.000287$ & $0.000649$ & $0.001141$\\ 
$(9,8)$ & $0.000440$  & $0.000237$ & $0.000604$ & $0.000976$\\ 
$(10,1)$ & $0.003322$  & $0.003322$ & $0.001107$ & $0.009093$\\ 
$(10,2)$ & $0.002746$  & $0.001445$ & $0.000964$ & $0.004482$\\ 
$(10,3)$ & $0.002302$  & $0.000865$ & $0.000865$ & $0.002943$\\ 
$(10,4)$ & $0.001929$  & $0.000593$ & $0.000788$ & $0.002173$\\ 
$(10,5)$ & $0.001595$  & $0.000439$ & $0.000726$ & $0.001710$\\ 
$(10,6)$ & $0.001286$  & $0.000341$ & $0.000673$ & $0.001400$\\ 
$(10,7)$ & $0.000987$  & $0.000275$ & $0.000627$ & $0.001178$\\ 
$(10,8)$ & $0.000686$  & $0.000228$ & $0.000586$ & $0.001011$\\ 
$(10,9)$ & $0.000369$  & $0.000192$ & $0.000549$ & $0.000880$\\ 

$(11,1)$ & $0.003145$  & $0.003145$ & $0.001048$ & $0.009229$\\ 
$(11,2)$ & $0.002628$  & $0.001374$ & $0.000916$ & $0.004555$\\ 
$(11,3)$ & $0.002232$  & $0.000824$ & $0.000824$ & $0.002996$\\ 
$(11,4)$ & $0.001901$  & $0.000566$ & $0.000754$ & $0.002215$\\ 
$(11,5)$ & $0.001609$  & $0.000420$ & $0.000697$ & $0.001746$\\ 
$(11,6)$ & $0.001341$  & $0.000327$ & $0.000648$ & $0.001433$\\ 
$(11,7)$ & $0.001087$  & $0.000264$ & $0.000606$ & $0.001209$\\ 
$(11,8)$ & $0.000837$  & $0.000219$ & $0.000568$ & $0.001040$\\ 
$(11,9)$ & $0.000584$  & $0.000185$ & $0.000534$ & $0.000908$\\ 
$(11,10)$ & $0.000314$  & $0.000159$ & $0.000503$ & $0.000802$\\ 

$(12,1)$ & $0.002987$  & $0.002987$ & $0.000996$ & $0.009343$\\ 
\hline
\end{tabular}
\begin{tabular}{|c|c|c|c|c|}
\hline
$(n,s)$& $\HF_{k,s}$ & $\A_{k,s}$ & $\B_{k,s,w}$ & $\CF_{k,s,w}$\\
\hline
$(12,2)$ & $0.002518$  & $0.001309$ & $0.000873$ & $0.004617$\\ 
$(12,3)$ & $0.002161$  & $0.000788$ & $0.000788$ & $0.003040$\\ 
$(12,4)$ & $0.001865$  & $0.000542$ & $0.000722$ & $0.002251$\\ 
$(12,5)$ & $0.001605$  & $0.000403$ & $0.000669$ & $0.001778$\\ 
$(12,6)$ & $0.001368$  & $0.000315$ & $0.000624$ & $0.001461$\\ 
$(12,7)$ & $0.001146$  & $0.000254$ & $0.000585$ & $0.001234$\\ 
$(12,8)$ & $0.000932$  & $0.000211$ & $0.000551$ & $0.001064$\\ 
$(12,9)$ & $0.000720$  & $0.000179$ & $0.000519$ & $0.000931$\\ 
$(12,10)$ & $0.000504$  & $0.000154$ & $0.000491$ & $0.000824$\\ 
$(12,11)$ & $0.000272$  & $0.000134$ & $0.000465$ & $0.000736$\\ 
$(13,1)$ & $0.002846$  & $0.002846$ & $0.000949$ & $0.009440$\\ 
$(13,2)$ & $0.002417$  & $0.001251$ & $0.000834$ & $0.004670$\\ 
$(13,3)$ & $0.002092$  & $0.000755$ & $0.000755$ & $0.003079$\\ 
$(13,4)$ & $0.001823$  & $0.000521$ & $0.000694$ & $0.002282$\\ 
$(13,5)$ & $0.001589$  & $0.000388$ & $0.000644$ & $0.001804$\\ 
$(13,6)$ & $0.001378$  & $0.000303$ & $0.000602$ & $0.001485$\\ 
$(13,7)$ & $0.001181$  & $0.000245$ & $0.000566$ & $0.001256$\\ 
$(13,8)$ & $0.000993$  & $0.000204$ & $0.000533$ & $0.001085$\\ 
$(13,9)$ & $0.000810$  & $0.000173$ & $0.000505$ & $0.000950$\\ 
$(13,10)$ & $0.000628$  & $0.000148$ & $0.000478$ & $0.000843$\\ 
$(13,11)$ & $0.000440$  & $0.000129$ & $0.000454$ & $0.000755$\\ 
$(13,12)$ & $0.000237$  & $0.000114$ & $0.000432$ & $0.000681$\\ 
$(14,1)$ & $0.002720$  & $0.002720$ & $0.000907$ & $0.009486$\\ 
$(14,2)$ & $0.002324$  & $0.001199$ & $0.000799$ & $0.004705$\\ 
$(14,3)$ & $0.002026$  & $0.000725$ & $0.000725$ & $0.003102$\\ 
$(14,4)$ & $0.001780$  & $0.000501$ & $0.000667$ & $0.002307$\\ 
$(14,5)$ & $0.001567$  & $0.000373$ & $0.000621$ & $0.001827$\\ 
$(14,6)$ & $0.001375$  & $0.000292$ & $0.000581$ & $0.001511$\\ 
$(14,7)$ & $0.001198$  & $0.000237$ & $0.000547$ & $0.001252$\\ 
$(14,8)$ & $0.001031$  & $0.000197$ & $0.000517$ & $0.001106$\\ 
$(14,9)$ & $0.000870$  & $0.000167$ & $0.000490$ & $0.000974$\\ 
$(14,10)$ & $0.000712$  & $0.000144$ & $0.000466$ & $0.000865$\\ 
$(14,11)$ & $0.000552$  & $0.000125$ & $0.000443$ & $0.000768$\\ 
$(14,12)$ & $0.000387$  & $0.000111$ & $0.000422$ & $0.000699$\\ 
$(14,13)$ & $0.000209$  & $0.000099$ & $0.000403$ & $0.000635$\\ \hline

\end{tabular}
}
\end{center}
\end{table}

\end{document}